\documentclass[showkeys,preprint,amsmath,amssymb,fleqn]{revtex4}
\usepackage{graphicx,amsmath,amsfonts,amssymb,float,bm,epsfig,color}
\usepackage{mathrsfs}

\newtheorem{definition}{Definition}
\newtheorem{postulate}{Postulate}
\newtheorem{theorem}{Theorem}

\newtheorem{corollary}{Corollary}
\newtheorem{remark}{Remark}

\newenvironment{proof}[1][Proof]{\noindent\textbf{#1.} }{\ \rule{0.5em}{0.5em}}

\begin{document}
\title{The Role of Second Law of Thermodynamics  in Continuum Physics: A Muschik and Ehrentraut Theorem Revisited}
\author{V. A. Cimmelli}
    \email{vito.cimmelli@unibas.it}
    \affiliation{Department of Mathematics, Computer Science and Economics, University of Basilicata, Viale dell'Ateneo Lucano, 10, 85100, Potenza, Italy}
\author{P. Rogolino}
    \email{progolino@unime.it}
    \affiliation{Department of Mathematical and Computer Sciences, Physical Sciences and Earth Sciences, University of Messina, Viale F. Stagno d'Alcontres, 31, 98166, Messina, Italy}

\bigskip

\bigskip

\begin{abstract}
Second law of thermodynamics imposes that in any thermodynamic process the entropy production must be nonnegative. In continuum physics such a requirement is fulfilled by postulating the constitutive equations  which represent the material properties of the bodies in such a way that second law of thermodynamics is satisfied in arbitrary processes. Such an approach, first assumed in some pioneering papers by Coleman and Noll \cite{ColNol} and Coleman and Mizel \cite{ColMiz}, in practice regards second law of thermodynamics as a restriction on the constitutive equations, which must guarantee that any solution of the balance laws satisfies also the entropy inequality. As observed by Muschik and Ehrentraut \cite{MusEhr}, this is a useful operative assumption, but not a consequence of general physical laws. Indeed, a different point of view, which regards second law of thermodynamics as a restriction on the thermodynamic processes, i.e., on the solutions of the system of balance laws, is possible. This is tantamount to assume that there are solutions of the balance laws which satisfy the entropy inequality, and solutions which do not satisfy it. In order to decide what is the correct approach, Muschik and Erhentraut postulated an amendment to the second law, which makes explicit the evident but rather hidden assumption that in any point of the body the entropy production is zero if, and only if, this point is thermodynamic equilibrium. Then they proved that, given the amendment,  second law of thermodynamics is necessarily a restriction on the constitutive equations and not on the thermodynamic processes. In the present paper we revisit their proof, lighting up some geometric aspects which were hidden in Ref. \cite{MusEhr}. Moreover, we propose an alternative formulation of second law of thermodynamics which incorporates the amendment. In this way we make this important result more intuitive and easily accessible to a wide audience.
\end{abstract}

\keywords{second law of thermodynamics; dissipation principle; state space; balance laws; entropy inequality}

\bibliographystyle{unsrt}

\maketitle

\section{Introduction}\label{1}

\medskip

Let's $B$ a continuous body undergoing a thermomechanical transformation, whose evolution in the spacetime is ruled by the system of balance laws
  \begin{equation} 
  \label{6}
   U_{\beta,t} +  U_{\beta,j}v_{j}+\Phi^\beta_{k,k}=r_\beta,\,\,\,\,\, \beta= 1\dots \omega,
  \end{equation} 

 with $v_{j}$ as the components of the velocity field  on $B$ entering the total time derivative, $\Phi^\beta_k$  as the components of the flux of $U_\beta$, and $r_\beta$ as the production of  $U_\beta$ (for the sake of simplicity we assume that the supplies are zero).  Moreover, the symbols $f_{,t}$ and $f_{,j}$ mean the partial derivative of function $f$ with respect to time and to the spatial coordinate $x_{j},\,\,j=1,2,3$, respectively.
We suppose that  the fields $U_{\beta}$, the fluxes $\Phi^\beta_k$, and the productions  $r_\beta$ depend on $\omega$ unknown fields $z_\alpha(x_{j},\,t)$ and on their spatial derivatives $z_{\alpha,j}(x_{j},\,t)$. Then, suitable constitutive equations must be assigned for them.

In classical rational thermodynamics \cite{TRU,ColNol} the equations above are the balances of mass, linear momentum, angular momentum and energy, while in the extended non-equilibrium thermodynamic theories taking the fluxes as independent variables, the set of field equations  includes the balance laws for the independent fluxes,  \cite{Gra,JCL,MR, SELCIMJOU,SzuKovSim}.  

The solutions the system  \eqref{6} must obey second law of thermodynamics, which imposes that the local entropy production
\begin{equation}\label{5}
 \sigma^{\left(s\right)}=\rho s_{,t}+\rho s_{,j}v_{j}+J_{k,k} - \rho({{r}}/{\vartheta}),
\end{equation}
where $s$ is the specific entropy,  $J_k$ are  the components of the entropy flux, and $\vartheta$ the absolute temperature, is nonnegative whatever the thermodynamic process is \cite{ColNol,TRU}. 

In continuum physics the entropy (or dissipation) principle \cite{CimJouRugVan} constitutes a valuable tool in modeling material properties. Coleman and Noll were the first to formulate it as follows \cite{ColNol}: 

\textit{The constitutive equations, which characterize the material properties of continuous media, must be assigned in such a way that second law of thermodynamics is satisfied along arbitrary thermodynamic processes}. 

\medskip

These authors also proposed a rigorous mathematical procedure to exploit the requirement above, currently referred to as Coleman-Noll procedure  \cite{ColNol,TriPapCimMus}.   

It is worth observing that   the entropy principle, such as formulated by Coleman and Noll, is just an operative assumption but not a consequence of a general physical law.  Thus, in principle, nothing prevents to assume that second law of thermodynamics  restricts the thermodynamic  processes instead of the constitutive equations, by selecting those which actually can occur in nature, and those which cannot occur. 
In order to decide what is the correct approach, Muschik and Ehrentraut \cite{MusEhr} proposed the following amendment to the second law: 

\medskip

\textit{Except in equilibria, reversible process directions in state space do not exist}. 

\medskip

From the physical point of view the amendment expresses, in form of postulate, the physically evident but rather hidden assumption that in any point of a continuum  body the entropy production is zero if, and only if, this point is in thermodynamic equilibrium.
Muschik and Ehrentraut  proved that, under the validity of the amendment, second law of thermodynamics necessarily restricts the constitutive equations and not the thermodynamic processes. In this way, the classical Coleman--Noll approach follows by a rigorous proof.  

The present paper is motivated by  the observation that the important result illustrated above can be put in more general and accessible form within a geometric framework. 

To achieve that task, we use the results in Refs. \cite{DolFraRog,DolFraRog1}, where a geometric perspective on nonequilibrium thermodynamics has been given. The chosen state space is different with respect to that considered Ref. \cite{MusEhr}, because we do not include in it the time derivatives. In this way, the constitutive equations we are dealing with, are suitable to respect the principle of material indifference, too   \cite{TRU}. After defining the space of the higher derivatives, we introduce the definitions of real, ideal,  and over-ideal vector of the higher derivatives. For thermodynamic processes, we give the definitions of irreversible, reversible, and over-reversible process, by analyzing the properties of its representative  curve in the fibre bundle of the configuration spaces.

Once the geometric framework is complete, we reformulate second law of thermodynamics, both locally and globally in time, in order to encompass  the amendment. In this way, we are able to prove a new formulation of the Muschik and Ehrentraut theorem.
The paper runs as follows.

In Sec. \ref{2}, we construct a new thermodynamic framework for nonequilibrium processes. In Sec. \ref{3}, we give a new formulation, both locally and globally, of second law of thermodynamics.  In Sec. \ref{4}, we prove the Muschik and Ehrentraut theorem. In Sec. \ref{5}, we resume our results and discuss some open problems which will be considered in future researches.

\section{The thermodynamic framework} \label{2}

\medskip

In this section we aim at constructing the geometric framework where our main result can be formulated. To this end, we start by giving some basic definitions.

\begin{definition} \label{d2}
The space $C_t$ of the configurations at the instant $t$ is represented by
a $\omega$-dimensional vector space  spanned by the solutions $z_\alpha(x_{j},\,t)$ of Eqs. \eqref{6} with a  the structure of a finite-dimensional manifold.
\end{definition}

We assume that the total configuration space is given by the disjoint union  
\begin{equation} 
 \label{7}
\mathcal{C} = \bigcup_{t\in [0,\,\infty]}\{t\} \times C_t,
\end{equation} 
with a given natural structure of a fibre bundle over the real line $\mathbb{R}$ where time flows \cite{DolFraRog,DolFraRog1}. 
\begin{definition} \label{d2}
$\mathcal{C}$ is called  configuration bundle.
\end{definition}
 Under the natural  assumption that  $C_t$ 
does not vary in time, namely,  $C_t=C\, \forall t$, then $\mathcal{C}$ has the topology of the
Cartesian product
\begin{equation} 
 \label{8}
\mathcal{C} =  \mathbb{R} \times C.
\end{equation}
\begin{definition} \label{d3}
A vector valued function $\pi: t\in [\tau_0,\,\tau_0+\tau] \subseteq\mathbb{R}\rightarrow z_\alpha(x_{j},\,t)  \in\mathcal{C}$ is said a thermodynamic process $\pi$ of duration $\tau$. Moreover,  $\pi=\pi(t)$ is  the parametric equation of the curve $\Gamma$ representative of $\pi$  in $\mathcal{C}$. 
\end{definition}
\begin{definition} \label{d31}
For $t_0 \in [\tau_0,\,\tau_0+\tau] $, a vector valued function $p: t\in [t_0,\,\tau_0+\tau] \subseteq\mathbb{R}\rightarrow z_\alpha(x_{j},\,t)  \in\mathcal{C}$ is said a restricted thermodynamic process $p$ of initial point $t_0$ and duration $\tau_0+\tau-t_0$, \cite{DolFraRog}. Moreover,  $p=p(t)$ is  the parametric equation of the curve $\gamma$ representative of $p$  in $\mathcal{C}$. 
\end{definition}

\begin{remark} \label{r0}
For $t_0= \tau_0$ we get $p(t)= \pi(t)$, for  $t_0= \tau_0+\tau$, $p(t)$ is the process of duration $0$, i.e., the null process.
\end{remark}

As said in Sec.  \ref{1}, in order to find the fields $z_\alpha(x_{j},\,t)$, i.e. to solve the system \eqref{6}, 
for  the quantities $U_\beta$, $\Phi^\alpha_k$ and $r_\alpha$  constitutive equations must be assigned on a suitable state space. 

\begin{definition} \label{d4}
The $4\omega$-dimensional vector space  with  the structure of a finite-dimensional manifold
\begin{equation}  
  \label{9}
   \Sigma_t = \left\{z_{\alpha}(x_{j},\,t),\,z_{\alpha,j}(x_{j},\,t)\right\}.
\end{equation} 
for any value of the time variable $t$, represents a local in time state space and it is called state space  at the instant $t$.
\end{definition}
\begin{definition} \label{d5}
The disjoint union
\begin{equation} 
 \label{10}
\mathcal{S} = \bigcup_{t\in [0,\,\infty]}\{t\} \times \Sigma_t,
\end{equation} 
with a given natural structure of a fibre bundle over the real line $\mathbb{R}$ where time flows,  represents 
the total configuration space and it is said the thermodynamic bundle.
\end{definition}
 Again, under the natural  assumption   that  $\Sigma_t$ 
does not vary in time, namely,  $\Sigma_t=\Sigma\, \forall t$, then $\mathcal{S}$ has the topology of the
Cartesian product
\begin{equation} 
 \label{11}
\mathcal{S} =  \mathbb{R} \times \Sigma.
\end{equation}
Of course, 
\begin{equation} 
 \label{12}
{C_t}  \subset \Sigma_t,\,\,\,\,\,\,\mathcal{C}\subset\mathcal{S}.
\end{equation}
The balance equations  \eqref{6} on the local in time state space $\Sigma_t$ read
 \begin{equation} 
 \label{15}
  \frac{\partial  U_{\beta}
}{\partial z_\alpha} z_{\alpha,t}+\frac{\partial  U_{\beta}
}{\partial z_{\alpha,j}} z_{\alpha,jt} + \frac{\partial  U_{\beta}
}{\partial z_\alpha} z_{\alpha,j}v_j +  \frac{\partial  U_{\beta}
}{\partial z_{\alpha,k}} z_{\alpha,kj}v_j  + \frac{\partial \Phi^\beta _k}{\partial z_\alpha}z_{\alpha,k} + \frac{\partial \Phi^\beta _k}{\partial z_{\alpha,j}}z_{\alpha,jk}
             = r_\beta.
 \end{equation} 
  In Eqs.  \eqref{15} and we may individuate the $10\omega$ higher derivatives $\left\{z_{\alpha,t},\, z_{\alpha,jt},\, z_{\alpha,jk} \right\}$, which  are the space and time derivatives of the elements of $\Sigma_t$. 
 
 \begin{definition} \label{d6}
The local in time $10\omega$-dimensional vector space 
\begin{equation}  
  \label{17}
   H_t = \left\{z_{\alpha,t}(x_{j},\,t),\,  z_{\alpha,jt}(x_{j},\,t),\,z_{\alpha,jk}(x_{j},\,t)\right\}, 
\end{equation} 
and the fibre bundle
\begin{equation} 
 \label{18}
\mathcal{H} = \bigcup_{t\in [0,\,\infty]}\{t\} \times H_t.
\end{equation}
represent the space of the higher derivatives at time $t$ and its
 fibre bundle respectively.
Moreover, the equilibrium subspace of $H_t$ and its fibre bundle are given by 
\begin{equation}  
  \label{17a}
   \hat E_t = \left\{z_{\alpha,jk}(x_{j},\,t)\right\}, 
\end{equation} 
and 
\begin{equation} 
 \label{18a}
\mathcal{\hat E} = \bigcup_{t\in [0,\,\infty]}\{t\} \times \hat E_t.
\end{equation} 
\end{definition}
  Analogously, the entropy inequality on the state space reads
\begin{equation} 
 \label{16}
  \rho\frac{\partial s}{\partial z_\alpha} z_{\alpha,t}+\rho\frac{\partial s}{\partial z_{\alpha,j}} z_{\alpha,jt} +  \rho\frac{\partial s}{\partial z_\alpha} z_{\alpha,j}v_j +  \rho\frac{\partial s}{\partial z_{\alpha,k}} z_{\alpha,kj}v_j 
 +\frac{\partial J_k}{\partial z_\alpha}z_{\alpha,k} +\frac{\partial J_k}{\partial z_{\alpha,j}}z_{\alpha,jk}  \geq0.
\end{equation} 

\begin{definition} \label{d7}
The local in time $10\omega$-dimensional vector space at time $t$ 
\begin{equation}  
  \label{17a}
   W_t = \left\{z_{\alpha,t}(x_{j},\,t),\,  z_{\alpha,jt}(x_{j},\,t),\,z_{\alpha,jk}(x_{j},\,t)\right\}, 
\end{equation} 
and the fibre bundle
\begin{equation} 
 \label{18}
\mathcal{W} = \bigcup_{t\in [0,\,\infty]}\{t\} \times W_t.
\end{equation} 
define the vector space and the fibre bundle of the higher derivatives, respectively, whose state vectors satisfy  the entropy inequality. Moreover, 
the equilibrium subspace of $W_t$ and its fibre bundle are given by 
\begin{equation}  
  \label{17a}
   E_t = \left\{z_{\alpha,jk}(x_{j},\,t)\right\}, 
\end{equation} 
and 
\begin{equation} 
 \label{18a}
\mathcal{E} = \bigcup_{t\in [0,\,\infty]}\{t\} \times  E_t.
\end{equation} 
\end{definition}
\medskip

\begin{remark}  \label{r1}
The reason because we defined two different spaces of the higher derivatives, one for the balance equations and another one for the entropy inequality, is related to the fundamental focus of the present investigation, namely, to determine the conditions, if any, under which all the solutions of the balance laws are also solutions of the entropy inequality. This will be discussed in detail in next section.
\end{remark} 

The relations in  Eqs.  \eqref{15} and \eqref{16} can be arranged as follows
  \begin{equation} 
 \label{19}
  \frac{\partial  U_{\beta}
}{\partial z_\alpha} z_{\alpha,t}+\frac{\partial  U_{\beta}
}{\partial z_{\alpha,j}} z_{\alpha,jt} +  \left(\frac{\partial  U_{\beta}
}{\partial z_{\alpha,k}} v_j  + \frac{\partial \Phi^\beta _j}{\partial z_{\alpha,k}}\right)z_{\alpha,kj} = r_\beta - \frac{\partial  U_{\beta}
}{\partial z_\alpha} z_{\alpha,j}v_j  -\frac{\partial \Phi^\beta _j}{\partial z_\alpha}z_{\alpha,j}.
\end{equation} 
\begin{equation} 
 \label{20}
  \rho\frac{\partial s}{\partial z_\alpha} z_{\alpha,t}+\rho\frac{\partial s}{\partial z_{\alpha,j}} z_{\alpha,jt} + \left(\rho\frac{\partial s}{\partial z_{\alpha,k}}v_i 
  +\frac{\partial J_i}{\partial z_{\alpha,k}} \right)z_{\alpha,ki} \geq -  \rho\frac{\partial s}{\partial z_\alpha} z_{\alpha,j}v_j -\frac{\partial J_i}{\partial z_\alpha}z_{\alpha,i}.
\end{equation} 
Let's now define the $10\omega \times 1$  column vector function
 \begin{equation} 
 \label{21}
y_{\alpha}  \equiv  \left( z_{\alpha,t},\,\, z_{\alpha,jt},\,\,z_{\alpha,kj}\right)^T, 
 \end{equation}
the $\omega \times 1$ column vector 
  \begin{equation} 
 \label{22}
 C_\beta \equiv  r_\beta - \frac{\partial  U_{\beta}
}{\partial z_\alpha} z_{\alpha,j}v_j  -\frac{\partial \Phi^\beta _j}{\partial z_\alpha}z_{\alpha,j},\,\,\,\, \beta =1 \dots \omega,
 \end{equation}
and the $\omega \times 10\omega$ matrix
 \begin{equation} 
 \label{23}
 A_{\beta\alpha} \equiv \left[ \frac{\partial  U_{\beta}
}{\partial z_\alpha},\,\,\,\,\frac{\partial  U_{\beta}}{\partial z_{\alpha,j}},\,\,\,\,\left(\frac{\partial  U_{\beta}
}{\partial z_{\alpha,k}} v_j  + \frac{\partial \Phi^\beta _j}{\partial z_{\alpha,k}}\right)\right]\,\,\,\,\,(j,\,k=1,2,3),
\end{equation}
with $C_\beta$ and  $A_{\beta\alpha}$ defined on $\mathcal{S}$. 
In this way, the balance equations \eqref{19} can be rearranged as
\begin{equation} 
 \label{24}
 A_{\beta\alpha}(\mathcal{S})y_{\alpha} = C_\beta(\mathcal{S}).
 \end{equation}
 Analogously, after defining the $10\omega \times 1$  column vector function
\begin{equation} 
 \label{25}
 B_{\alpha} (\mathcal{S}) \equiv \left(\rho\frac{\partial s}{\partial z_{\alpha}},\,\,\,\, \rho\frac{\partial s}{\partial z_{\alpha,j}},\,\,\,\,  \left(\rho\frac{\partial s}{\partial z_{\alpha,k}}v_i 
  +\frac{\partial J_i}{\partial z_{\alpha,k}} \right)\right)^T, 
\end{equation} 
and the scalar function
\begin{equation} 
 \label{26}
D(\mathcal{S}) \equiv -  \rho\frac{\partial s}{\partial z_\alpha} z_{\alpha,j}v_j -\frac{\partial J_i}{\partial z_\alpha}z_{\alpha,i},
\end{equation} 
we can write the inequality \eqref{20} as
\begin{equation} 
 \label{27}
 B_{\alpha}(\mathcal{S})y_{\alpha} \ge D(\mathcal{S}).
 \end{equation}

\begin{remark}  \label{r2}
It is worth observing that the higher derivatives entering the system  \eqref{24} are elements of $H_t$, while those entering the inequality \eqref{27} are elements of $W_t$.
\end{remark} 

From now on we pursue our analysis under the hypothesis that $B$ occupies the whole space.  Then, for arbitrary $t_0 \in [\tau_0, \tau_0+\tau]$ we  consider the restricted process $p$ of initial instant $t_0$ and duration $\tau_0+\tau-t_0$, and suppose that it corresponds to the solution of the Cauchy problem  for the system \eqref{24} with  initial conditions
\begin{equation} 
 \label{28}
 z_{\alpha}(x_{j},\,t_0) =  z_{\alpha\,0}({x}_{j}),\,\,\,\, \forall {P} \in C.
 \end{equation}
 If $A_{\beta\alpha}$ and $C_\beta$ are regular, and $A_{\beta\alpha}$ is invertible,  the theorem of Cauchy-Kovalevskaya ensures that  the Cauchy problem   \eqref{24} and \eqref{28} has a unique solution continuously depending on the initial data \eqref{28}, \cite{CH}. However, such a solution does not necessarily corresponds to a thermodynamic process which is physically realizable, since the physically admissible solutions of  
 \eqref{24} and \eqref{28} are only those solutions which additionally satisfy the unilateral differential constraint \eqref{27}. On the other hand, the problem \eqref{24} and \eqref{28} is very difficult to solve, in general, so that to find a solution of it and verify ex post if it also satisfies \eqref{27} does not seems to be a convenient procedure. For that reason, Coleman and Noll \cite{ColNol} in 1963 postulated the  constitutive principle referred in Sec. \ref{1},   \cite{CimJouRugVan}.  Then it is important to investigate if  the Coleman and Noll postulate is a consequence of a general physical law or it is an arbitrary, although very useful, assumption, as observed by Muschik and Ehrentraut \cite{MusEhr}. Such a study will be carried on in the next sections.
 
\section{Local and global formulation of second law of thermodynamics} \label{3}

\medskip

Let's consider now a fixed point $P_0 \in B$ whose vector position will be indicated by $\mathbf{x}_0$, a fixed instant of time $t_0 \in [\tau_0, \, \tau_0+\tau]$.  We note that, whatever is $t_0$,  it can  can ever be considered as the initial time of a restricted process of duration $\tau_0+\tau-t_0$. Moreover, let $\Sigma_0$,  $H_0$, and $E_0$ the vector spaces 
 $\Sigma_t(P_0,\,t_0)$, $H_t(P_0,\,t_0)$, and $E_t(P_0,\,t_0)$. When evaluated in $(P_0,\,t_0)$, the balance equations \eqref{24} and the entropy inequality \eqref{27} transform in the algebraic relations
 \begin{equation}
 \label{29}
 A_{\beta\alpha}({\Sigma_0})y_{\alpha} = C_\beta({\Sigma_0}),
 \end{equation}
 \begin{equation} 
 \label{30}
 B_{\alpha}({\Sigma_0})y_{\alpha} \ge D({\Sigma_0}).
 \end{equation}
In this way we can regard the $\omega \times 10\omega$ matrix $A_{\beta\alpha}({\Sigma_0})$ as a linear morphism from  $H_0$ to the $\omega$-dimensional Euclidean vector space defined on $\Sigma_0$. Analogously, the vector $B_{\alpha}({\Sigma_0})$ can be regarded as a linear application from $H_0$ in $\mathbb{R}$, so that $B_{\alpha}({\Sigma_0})$ belongs to the dual space $H^ {*}_0$ of  $H_0$. It is worth observing that, since $A_{\beta\alpha}$ has been supposed to be invertible (otherwise the Cauchy problem  \eqref{24} and \eqref{28} would not admit a unique solution), the algebraic relations  \eqref{29} allow to determine $\omega$ of the $10\omega$ components of $y_{\alpha}$.  Moreover, by spatial derivation  of the initial conditions  \eqref{28} we get
\begin{equation} 
 \label{31}
  z_{\alpha,jk}(x_{j},\,t_0) =  z_{\alpha\,0,jk}({x_j}), 
 \end{equation}
which, once evaluated in $P_0$, allow to determine  $6\omega$ components of $y_{\alpha}$. It is worth observing that, since the initial conditions can be assigned arbitrarily, such $6\omega$ quantities can assume arbitrary values. Moreover, there are further $3\omega$ components of the vector $y_{\alpha}$ which remain completely arbitrary, since the system \eqref{29}   and the initial relations  \eqref{31} allow to determine only  $7\omega$ of the  $10\omega$ components of $y_{\alpha}$.  Then, it is not guaranteed that the inequality  \eqref{30} is satisfied whatever is $y_{\alpha} \in H_0$.  Thus, we define the space $W_0 \subseteq H_0$ constituted by the vectors of $H_0$ which satisfy both Eq.  \eqref{29} and the inequality  \eqref{30}.

\begin{remark}  \label{r3}
It is worth observing that, although it is not guaranteed that the inequality  \eqref{30} is satisfied whatever is $y_{\alpha} \in H_0$, at this stage we do not have elements to exclude such a possibility. In other words, we do not have elements to decide if, actually,  $W_0$ is a proper subspace of $H_0$ or it coincides with $H_0$.
\end{remark} 

In order to decide if $W_0 \subset H_0$, or $W_0=H_0$, we follow the way paved by Muschik and Ehrentraut \cite{MusEhr} who observed that such a decision cannot ensue by the sole second law of thermodynamics, because such a law does not contain information neither regarding Eqs.  \eqref{29}, nor regarding the initial conditions \eqref{31}. In order to fill this gap, Muschik and Ehrentraut completed the information contained into the inequality \eqref{30} by an amendment which clarifies how the reversible transformations can be realized from the operative point of view. Here we follow their strategy, but propose a more general approach which includes the amendment into a new formulation of second law. To achieve that task, we need some preliminary definitions. To this end, we observe that in the real world reversible thermodynamic transformations do not exist, but they are approximated by very slow (quasi-static) transformations in which in any point $P_0 \in B$ the system is very close to the thermodynamic equilibrium. From a ideal point of view, a quasi-static transformation requires an infinite time to occur, and in any point of the system the value of the state  variable is constant in time. 

\begin{remark}  \label{r4}
Alternative formulations of the thermodynamic laws which consider realistic transformations occurring in a finite time have been proposed within the framework of finite time thermodynamics  \cite{AndBerNitSal,AndSalBer,Hof}.
\end{remark} 

As far as the thermodynamic framework developed so far is concerned,  if $B$ undergoes a quasi-static transformation, along with Muschik and Ehrentraut \cite{MusEhr}, we say that in any point $(P_0,\,t_0)$ the vectors of the higher derivatives are elements of $E_0$. Such an observation suggests the following definitions.

\begin{definition} \label{d8}
 A vector $y_{\alpha} \in H_0$ is said:

\medskip

\begin{itemize} 
\item real, if it satisfies the relation $ B_{\alpha}({\Sigma_0})y_{\alpha} > D({\Sigma_0})$;
\smallskip
\item ideal, if it satisfies the relation $ B_{\alpha}({\Sigma_0})y_{\alpha} = D({\Sigma_0})$;
\smallskip
\item over-ideal, if it satisfies the relation $ B_{\alpha}({\Sigma_0})y_{\alpha} < D({\Sigma_0})$.
\end{itemize} 
\end{definition}
\medskip
  
Owing to the definitions above we can establish the following 

\begin{postulate}{\bf Local formulation of second law of thermodynamics}. \label{p1}
Let $B$ a body, and let the couple $(P_0,\,t_0)$ represent an arbitrary point of $B$ at an arbitrary instant $t_0 \in [\tau_0,\, \tau_0+\tau]$. Suppose $B$ is  undergoing an arbitrary  thermodynamic process of initial instant $t_0$ and duration $\tau_0+\tau - t_0$.   Then,  the local space of the higher derivatives $W_0$ does not contain over-ideal vectors.  Moreover, a vector $y_{\alpha} \in W_0$ is  ideal if, and only if,  $y_{\alpha} \in E_0$.
\end{postulate}
The postulate above traduces the experimental evidence that in a thermodynamic process the entropy production cannot be negative in any point $P_0$ of $B$ at any instant $t_0$. Moreover, it also expresses the further experimental fact, which is often tacit in the formulations of second law of thermodynamics, that the entropy production can be zero only in the points of $B$ which are in equilibrium. In particular, we say that the point $P_0 \in B$ at the instant $t_0$  is in thermodynamic equilibrium if, and only if, $W_0= E_0$.

\medskip

\begin{remark}  \label{r5}
We note that the local formulation of second law of thermodynamics prohibits that over-ideal vectors are in $W_0$ but does not prevents they are in $H_0$. If $H_0$ contains over-ideal vectors or not is just the focus of the present investigation.
\end{remark} 

 \begin{definition} \label{d9}
 Let $B$ a body undergoing an arbitrary  thermodynamic process $p$ of initial instant $t_0$ and duration $\tau_0+\tau - t_0$, and let $\gamma$ the curve representative of the process in  $\mathcal{C}$.
The process $p$ is said:
 \medskip
 \begin{itemize} 
\item irreversible, if there exists at least a point  $z_{\alpha}(P,\,t)$ of $\gamma$ in which the vector of the higher derivatives $y_{\alpha}(P,\,t)$ is real;
\smallskip
\item reversible,  if in any point $z_{\alpha}(P,\,t)$ of $\gamma$  the vector of the higher derivatives $y_{\alpha}(P,\,t)$ is ideal;
\smallskip
\item over-reversible, if there exists at least a point $z_{\alpha}(P,\,t)$ of $\gamma$ in which the vector of the higher derivatives $y_{\alpha}(P,\,t)$ is over-ideal.
\end{itemize}
\end{definition} 

The definitions above allow to enunciate the following

\begin{postulate} \label{p2}  
{\bf Global formulation of second law of thermodynamics}:
Over-reversible processes do not occur in nature. Moreover, a thermodynamic process is reversible if, and only if, its representative curve $\gamma$  lies  into the equilibrium bundle  $\mathcal{E}$.
\end{postulate}
 
 The previous formulations (local and global)  of second law of thermodynamics include the information, not present in the classical ones, that
 the reversible transformations are necessarily quasi-static and hence, they need an infinite time to occur. So, they represent ideal processes, which in nature are approximated by very slow transformations. Here we take into account such a situation by admitting that in any point of a reversible curve the vector of the higher derivatives is ideal.

\section{The Muschik and Ehrentraut theorem revisited} \label{4}

\medskip

In this section we present a novel formulation of the Muschik and Ehrentraut theorem proved in Ref. \cite{MusEhr}. To this end, we use the thermodynamic framework and the generalized formulations of second law established above.

\begin{theorem} \label{t1}
Let $B$ a body, and let the couple $(P_0,\,t_0)$ represent an arbitrary point of $B$ at an arbitrary instant $t_0 \in [\tau_0,\, \tau_0+\tau]$.  Then,  $H_0$ = $W_0$.
 \end{theorem}
 \begin{proof}
To prove the theorem it is enough to demonstrate that the vectors of $H_0$ are all and only the vectors of $W_0$.
To this end, we observe that, in the generic point $(P_0,\,t_0)$, at fixed values of  $A_{\beta\alpha}({\Sigma_0})$,$C_\beta({\Sigma_0})$,   $B_{\alpha}({\Sigma_0})$, and $D({\Sigma_0})$, correspond infinite vectors $y_{\alpha}(P_0,\,t_0)$, because only $\omega$ components of $y_{\alpha}(P_0,\,t_0)$ are determined by the balance equations while the remaining $9\omega$ are completely arbitrary (see discussion in Sec. \ref{3}). Moreover, if all the $y_{\alpha}$ in $H_0$ would be over-ideal,   the vector space $W_0$ would be empty, because the second law of thermodynamics prohibits that it contains over-ideal vectors. As a consequence, in $(P_0,\,t_0)$ no any process would be possible. On the other hand, since $(P_0,\,t_0)$ is arbitrary, no any thermodynamic transformation could occur in $B$ in the interval of time $[\tau_0,\,\tau_0+\tau]$. So, in $(P_0,\,t_0)$ the space $H_0$ contains, in principle, both real/ideal vectors and over-ideal ones. 
 
Let's suppose that in $(P_0,\,t_0)$ the space $H_0$ contains an ideal vector $y^1_{\alpha}$  and an over-ideal vector $y^2_{\alpha}$. Since the existence of $y^1_{\alpha}$ is possible if, and only if, $(P_0,\,t_0)$ is in thermodynamic equilibrium, while $y^2_{\alpha}$ exist if, and only if $(P_0,\,t_0)$ is not  in thermodynamic equilibrium, such a situation is impossible to be realized. 

Analogously, let's suppose that $y^1_{\alpha}$ is ideal  and $y^2_{\alpha}$ is real. Again, such a situation is impossible, because it would require $(P_0,\,t_0)$ to be in equilibrium and not in equilibrium.

 Finally, let $y^1_{\alpha}$ be a real vector, and  $y^2_{\alpha}$ a over-ideal one. Such a situation is possible, in principle, provided   
 $(P_0,\,t_0)$ is not in equilibrium.

In such a case, due to the local formulation of second law, neither $y^1_{\alpha}$ nor $y^2_{\alpha}$ are elements of $E_0$.
 
Let's consider now the linear combination $y^3_{\alpha}= \lambda y^1_{\alpha} + (1-\lambda)y^2_{\alpha}$, with $\lambda \in ]0,\,1[$. Since $y^1_{\alpha}$ and $y^2_{\alpha}$ are in $H_0$, they satisfy the following equations

 \begin{equation}
 \label{32}
 A_{\beta\alpha}({\Sigma_0})y^1_{\alpha} = C_\beta({\Sigma_0}),
 \end{equation}
 \begin{equation}
 \label{33}
 A_{\beta\alpha}({\Sigma_0})y^2_{\alpha} = C_\beta({\Sigma_0}).
 \end{equation}
The combination of Eqs. \eqref{32} multiplied by  $\lambda$ and Eqs. \eqref{33} multiplied by $(1-\lambda)$ leads to
 \begin{equation}
 \label{34}
 A_{\beta\alpha}({\Sigma_0})y^3_{\alpha} = C_\beta({\Sigma_0}),
 \end{equation}
namely, $y^3_{\alpha}$ is also a solution of  Eq \eqref{29}, i.e. it is in $H_0$.  On the other hand, the local entropy production corresponding to $y^3_{\alpha}$ can be written as
 \begin{equation} 
  \label{35}
\sigma^3=\lambda\left[B_{\alpha}({\Sigma_0})y^1_{\alpha} - D({\Sigma_0})\right] + (1-\lambda)\left[B_{\alpha}({\Sigma_0})y^2_{\alpha} - D({\Sigma_0})\right]=
 \end{equation}
 \begin{equation} 
 \nonumber
=B_{\alpha}({\Sigma_0})\left[\lambda y^1_{\alpha} + (1-\lambda)y^2_{\alpha}\right ]- D({\Sigma_0}).
 \end{equation}
Since $\lambda$ is arbitrary in $]0,\,1[$, nothing prevents to chose it as
 \begin{equation} 
 \label{36}
\lambda = \frac{D({\Sigma_0})-B_{\alpha}({\Sigma_0})y^2_{\alpha}}{B_{\alpha}({\Sigma_0})\left[y^1_{\alpha}-y^2_{\alpha}\right]},
 \end{equation}
because, as it is easily seen, the right-hand side of Eq. \eqref{36} is in the interval $]0,\,1[$. In fact, being $y^2_{\alpha}$ over-ideal we get $D({\Sigma_0})-B_{\alpha}({\Sigma_0})y^2_{\alpha}>0$. Moreover, being $y^1_{\alpha}$ real, we get $B_{\alpha}({\Sigma_0})\left[y^1_{\alpha} -y^2_{\alpha}\right ]> D({\Sigma_0})-D({\Sigma_0})$, namely, $B_{\alpha}({\Sigma_0})\left[y^1_{\alpha} -y^2_{\alpha}\right ]> 0$. Hence $\lambda>0$. Moreover, being $y^1_{\alpha}$ real, we get also that $B_{\alpha}({\Sigma_0})\left[y^1_{\alpha}-y^2_{\alpha}\right]> D({\Sigma_0})-B_{\alpha}({\Sigma_0})y^2_{\alpha}$, and hence
$\lambda <1$. 

Consequently, the right-hand side of Eq.  \eqref{35} vanishes, so that $y^3_{\alpha}$ is in $E_0$. However,  this is impossible, otherwise  $(P_0,\,t_0)$ would be in thermodynamic equilibrium. Thus, it is forbidden that in  $(P_0,\,t_0)$ there are both real and over-ideal vectors which are solutions of the local balance laws  \eqref{29}.

Furthermore, suppose that both  $y^1_{\alpha}$ and  $y^2_{\alpha}$ are real. Then, it is easy to verify by direct calculation that $\lambda$ can be taken such that $\sigma^3>0$. 

Finally, if $(P_0,\,t_0)$ is a point of equilibrium, then the entropy production related to  $y^1_{\alpha}$ and  $y^2_{\alpha}$ vanishes, so that, by Eq. \eqref{35}, it follows that also $\sigma^3$ is zero.

The considerations above show the impossibility that in a point  $P_0$ of $B$, at a given instant $t_0$, the solutions of Eqs. \eqref{29} can be of different type. Moreover, they cannot be over-ideal only, because this contradicts the local form of second law of thermodynamics. Thus, $H_0$ may contain either only real vectors, and in such a case $(P_0,\,t_0)$ is a point  of non-equilibrium, or only ideal vectors, and in such a case $(P_0,\,t_0)$ is a point of equilibrium. This conclusion proves the theorem.
\end{proof}

\begin{corollary} \label{c1}
 $\mathcal{H}$ = $\mathcal{W}$.
\end{corollary} 
\begin{proof} This corollary is an immediate consequence of the arbitrariness of the initial instant $t_0$, and of the point $P_0$. In particular, whatever is $t_0$, we can ever consider it as the initial instant of the restricted process of duration  $\tau_0+\tau-t_0$, so that $H_0 \equiv H_t(P_0,t_0)$ has dimension $10\omega$. Moreover, only $7\omega$ of components of the vectors of $H_0$ can be determined by the algebraic relations \eqref{29} and \eqref{31} while the further $3\omega$ components are completely arbitrary. Thus, to $H_0$ can be applied the conclusions established  in Theorem  \ref{1}. This is enough to prove that, for any $t \in [\tau_0,\tau_0+\tau]$ the space of the higher derivatives $H_t$ contains only real or ideal vectors.
\end{proof}

\begin{remark}  \label{r6}
The Corollary  \ref{c1} also implies  $\mathcal{\hat E}$ = $\mathcal{E}$.
\end{remark} 

\begin{corollary} \label{c2}
 The unilateral differential constraint  \eqref{27} is a restriction on the constitutive quantities $U_{\beta}$, $r_{\beta}$, $s$ and $J_k$ and not on the thermodynamic processes $p$.
 \end{corollary}
\begin{proof}
 In fact, any process $p: t\in [\tau_0,\,\tau_0+\tau] \subseteq\mathbb{R}\rightarrow z_\alpha(x_{j},\,t)  \in\mathcal{C}$, where $z_\alpha(x_{j},\,t)$ is a solution of the balance laws  \eqref{24}, can only be either irreversible or reversible but not over-reversible, because otherwise its representative curve $\gamma$ would contain at least a over-ideal point, against Corollary \ref{c1}. On the other hand, such a property of the solutions of the system of balance laws is not guaranteed whatever are $A_{\beta\alpha}$ and $C_\beta$, and for arbitrary $s$ and $J_k$ because, given the state space,  only particular forms of those functions defined on it lead to a nonnegative entropy production. Then, the role of the unilateral differential constraint in Eq. \eqref{27} is just to select such forms.
\end{proof}
\section{Discussion} \label{5}
Exploitation of second law of thermodynamics is based on the assumption that it restricts the
constitutive equations and not the thermodynamic processes. Then, the constitutive equations must be assigned in such a way that
 all solutions of the field equations satisfy the entropy inequality.
An alternative interpretation of the restrictions imposed by second law  is that we must exclude from the set of solutions
of the balance equations that ones which  do not guarantee a nonnegative entropy production. The
problem of choosing among the two interpretations above has been solved in 1996 by Muschik and Ehrentraut \cite{MusEhr},
by postulating an amendment to the second law which assumes that at a fixed instant of time and in any point of the body, the entropy production is zero if, and only if, this point is in thermodynamic equilibrium. Muschik and Ehrentraut  proved that, presupposing the amendment,
necessarily second law of thermodynamics restricts the constitutive equations and not the processes. Such
a result justifies, from the theoretical point of view, the approach to the exploitation of second law proposed in
1963 by Coleman and Noll in their celebrated paper \cite{ColNol}.

In the present paper we have revisited their proof, lighting up some geometric aspects which were hidden in Ref. \cite{MusEhr}. Moreover, we proposed a generalized formulation of second law of thermodynamics which incorporates the amendment. 

In future researches we aim at extending the present results to more complex situations.

In the case of shock wave propagation, among the solutions of the Rankine-Hugoniot equations, the physical
shocks are selected by the celebrated Lax conditions, which force the shock speed $U_s$ to satisfy the
inequality $U_{b}
> U_{s} > U_{a}$ with $U_{b}$ as the characteristic speed behind the shock and $U_{a}$ as the
characteristic speed ahead the shock  \cite{DAF}. Since for a fluid the Lax conditions imply the growth of the
entropy across the shock, they are often called in the literature ''entropy growth conditions'' .
The common interpretation of this result is that for non-regular (weak) solutions of the balance equations,
second law of thermodynamics restricts the processes instead of the constitutive equations. However, in Ref. \cite{TriCim} it is proved that the amendment can be generalized in order to prove that second law of thermodynamics necessarily
restricts the constitutive equations on both sides of the shock. So, under the hypothesis above, the classical
interpretation of Lax conditions should be revisited in the light of the new mathematical framework formulated in the present paper.

In Ref. \cite{TriCim1}, the results in Refs. \cite{MusEhr,TriCim} on the interpretation of the second law of thermodynamics have been extended in order to encompass the most general situation in which also the gradients of the basic laws are considered
as constraints for the entropy inequality \cite{CimSelTri}. This  result too should be reanalyzed within the mathematical framework presented here. 

To our opinion, the investigations mentioned above are necessary, since thermodynamic processes that involve discontinuous
solutions are very frequent in physics.

\section*{Acknowledgements}

P.~R. thanks the University of Messina and the Italian National Group of Mathematical Physics (GNFM-INdAM) for financial support.

V.~A.~C. thanks the University of Basilicata nd the Italian National Group of Mathematical Physics (GNFM-INdAM)  for financial support.

\bibliographystyle{10}

\begin{thebibliography}{99}
%
\bibitem{ColNol} B. D. Coleman, W. Noll, \textit{The thermodynamics of elastic materials with heat
conduction and viscosity}, Arch. Rational Mech. Anal., {13} (1963), 167--178.
%
\bibitem{ColMiz} B. D. Coleman, V. J. Mizel, \textit{Existence of caloric equations of state in thermodynamics}, 
J. Chem. Phys., {40} (1964), 1116--1125.
%
\bibitem{MusEhr}
W. Muschik, H. Ehrentraut, \textit{An Amendment to the Second Law}, J. Non-Equilib. Thermodyn. {21} (1996), 175--192.
%
\bibitem{TRU}
C.~Truesdell.
\newblock \textit{Rational Thermodynamics},
\newblock 2nd enlarged edition.
\newblock New York: Springer, 1984.
%
\bibitem{Gra}
H.~Grad,  \textit{On the kinetic theory of rarefied gases},  Comm. Pure Appl.
  Math., 2 (1949), 331--407, .
%
\bibitem{JCL}
D.~Jou, J.~Casas-V\'{a}zquez, G.~Lebon, \textit{Extended Irreversible
  Thermodynamics}.
\newblock Berlin: Springer, fourth revised~ed., 2010.
%
\bibitem{MR}
I.~M\"{u}ller, T.~Ruggeri,  \textit{Rational Extended Thermodynamics}.
\newblock New York: Springer, second~ed., 1998.
%
\bibitem{TriPapCimMus}
V. Triani, C. Papenfuss, V. A. Cimmelli, W. Muschik, \textit{Exploitation of the Second Law: Coleman-Noll
and Liu Procedure in Comparison
}, J. Non-Equilib. Thermodyn. {33} (2008), 47-60.
%
\bibitem{SELCIMJOU}
A.~Sellitto, V.~A. Cimmelli, D.~Jou, \textit{Mesoscopic theories of heat
  transport in nanosystems}.
\newblock Berlin: Springer, 2016. 
%
\bibitem{SzuKovSim}
M. Szucs,  R. Kovacs, S. Simic,  \textit{Open Mathematical Aspects of Continuum Thermodynamics: Hyperbolicity,
Boundaries and Nonlinearities}, {Symmetry}
  12 (2020),~1469 (26 pages).
%
\bibitem{CimJouRugVan}
V.~A. Cimmelli, D.~Jou, T.~Ruggeri,  P.~V\'an,   {Entropy {P}rinciple and {R}ecent {R}esults in {N}on-{E}quilibrium {T}heories}, {Entropy}
  16 (2014),~1756--1807.
%
\bibitem{DolFraRog}
M. Dolfin, M. Francaviglia, P. Rogolino, \textit{A Geometric Perspective on Irreversible
Thermodynamics with Internal Variables}, J. Non-Equilib. Thermodyn. {23} (1998), 250--263.
%
\bibitem{DolFraRog1}
M. Dolfin, M. Francaviglia, P. Rogolino, \textit{A geometric model for the thermodynamics of simple materials}, Periodica Polytechnica Ser. Mech. Eng. {43} (1999), 29--36.
%
\bibitem{CH}
R. Courant, D. Hilbert, \textit{Methods of Mathematical Physics: Partial
  Differential Equations, {V}ol.{I}{I}}.
\newblock New York: John Wiley and Sons, 1989.
%
\bibitem{AndBerNitSal}
B.~Andresen, R.~S. Berry, A.~Nitzan, P.~Salomon,  \textit{Thermodynamics in finite
  time. {I}. {T}he step-{C}arnot cicle}, {Phys. Rev. A}, 15 (1977), ~2086--2093.
%
\bibitem{AndSalBer}
B.~Andresen, P.~Salomon,  R.~S. Berry,  \textit{Thermodynamics in finite time},
  {Physics Today}, 37 (1984), ~62--70.
%
\bibitem{Hof}
K.~H. Hoffmann, \textit{Recent {D}evelopments in {F}inite {T}ime {Thermodynamics}},
  {TECHNISCHE MECHANIK}, ~22  2002,~14--25.
%
\bibitem{DAF}
C.~Dafermos, \textit{Hyperbolic conservation laws in continuum physics}.
\newblock 2nd  edition.
\newblock Berlin: Springer, 2009. 
%
\bibitem{TriCim}
V. Triani, V. A. Cimmelli, \textit{Interpretation of Second Law of Thermodynamics
in the presence of interfaces}, Continuum Mech. Thermodyn. {24} (2012), 165--174.
%
\bibitem{TriCim1}
V. Triani, V. A. Cimmelli, \textit{Entropy principle, non-regular processes, and generalized exploitation
procedures}, J. Math. Phys. {53} (2012), 063509 (8 pages).
%
\bibitem{CimSelTri}
V. A. Cimmelli, A. Sellitto,V. Triani, \textit{A new perspective on the form
of the first and second laws in rational thermodynamics:
Korteweg fluids as an example}, J. Non-Equilib. Thermodyn. {35} (2010), 251-265.
%
\end{thebibliography}

\end{document}